\documentclass[sigconf]{acmart}

\renewcommand\footnotetextcopyrightpermission[1]{}

\usepackage[ruled]{algorithm2e} 
\usepackage{algpseudocode}
\usepackage{amsmath,amsthm,mathtools}
\usepackage{ulem}
\usepackage{enumerate}
\usepackage{xcolor}
\usepackage{array,multirow,booktabs,graphicx}
\usepackage{float}
\usepackage{placeins}

\usepackage[shortlabels]{enumitem}
\usepackage{cleveref}
\usepackage{afterpage}
\usepackage{needspace}

\usepackage{balance}

\usepackage{placeins}  
\usepackage{afterpage}  

\crefname{algocf}{algorithm}{algorithms}
\Crefname{algocf}{Algorithm}{Algorithms}
\crefalias{algocf}{algorithm}

\DeclareMathOperator{\ve}{vec}
\DeclareMathOperator{\di}{diagMat}

\newcommand{\norm}[1]{\left\|#1\right\|}
\newcommand{\n}[1]{\left\|#1\right\|_F^2}
\renewcommand{\t}{^\top}
\newcommand{\Ii}{I_{|\mathcal{I}|}}

\DeclareRobustCommand{\parhead}[1]{\textbf{#1}~}

\AtBeginDocument{%
  }
\copyrightyear{2025}
\acmYear{2025}
\setcopyright{acmlicensed}\acmConference[WWW '25]{Proceedings of the ACM Web Conference 2025}{April 28-May 2, 2025}{Sydney, NSW, Australia}
\acmBooktitle{Proceedings of the ACM Web Conference 2025 (WWW '25), April 28-May 2, 2025, Sydney, NSW, Australia}
\acmDOI{10.1145/3696410.3714680}
\acmISBN{979-8-4007-1274-6/25/04}

\begin{document}

\title{Does Weighting Improve Matrix Factorization for Recommender Systems?}
\titlenote{This paper is published in the proceedings of 'The Web Conference' (WWW) 2025,
  under the Creative Commons Attribution 4.0 International (CC-BY 4.0) license.}

\author{Alex Ayoub}\authornote{This work was done while Alex Ayoub was at Netflix, Inc.}
\affiliation{%
  \institution{University of Alberta}
  \city{Edmonton, AB}
  \country{Canada}
}
\email{aayoub@ualberta.ca}

\author{Samuel Robertson}
\affiliation{%
  \institution{University of Alberta}
  \city{Edmonton, AB}
  \country{Canada}
}
\email{smrobert@ualberta.ca}

\author{Dawen Liang}
\affiliation{%
  \institution{Netflix, Inc.}
  \city{Los Gatos, CA}
  \country{USA}
}
\email{dliang@netfix.com}

\author{Harald Steck}
\affiliation{%
    \institution{Netflix, Inc.}
  \city{Los Gatos, CA}
  \country{USA}
}
\email{hsteck@netflix.com}

\author{Nathan Kallus}
\affiliation{%
  \institution{Netflix, Inc. \& Cornell University}
  \city{New York, NY}
  \country{USA}
}
\email{kallus@cornell.edu}

\begin{abstract}
Matrix factorization is a widely used approach for top-N recommendation and collaborative filtering. When implemented on implicit feedback data (such as clicks), a common heuristic is to upweight the observed interactions. This strategy has been shown to improve performance for certain algorithms. In this paper, we conduct a systematic study of various weighting schemes and matrix factorization algorithms. Somewhat surprisingly, we find that training with unweighted data can perform comparably to—and sometimes outperform—training with weighted data, especially for large models. This observation challenges the conventional wisdom. Nevertheless, we identify cases where weighting can be beneficial, particularly for models with lower capacity and specific regularization schemes. We also derive efficient algorithms for exactly minimizing several weighted objectives that were previously considered computationally intractable. Our work provides a comprehensive analysis of the interplay between weighting, regularization, and model capacity in matrix factorization for recommender systems.
\end{abstract}




\maketitle
\pagestyle{plain} 

\section{Introduction}

In a typical recommender system we observe how users interact with items, and our goal is to recommend each user previously unseen items that they would like. 
In this paper we consider making recommendations from \textit{implicit feedback}, which is the history of previous user interactions (e.g. clicks, purchases, views), as opposed to \textit{explicit feedback} (e.g. user-supplied ratings).
Implicit feedback is organically generated by users interacting with the system, and hence more easily accessible and prevalent \citep{hu2008collaborative}. 

Traditionally, matrix factorization has been the go-to choice for modeling implicit feedback data \citep{hu2008collaborative,pan2008,rendle2009bpr,steck2010,liang2016modeling}. At its core, matrix factorization involves decomposing the user-item interaction matrix into a product of two matrices, so as to capture latent patterns with a low-dimensional representation. Matrix factorization techniques are interesting because of their simplicity and robustness.
Even with the recent emergence of deep learning-based approaches, simple linear models like matrix factorization can, when carefully tuned, compete with and sometimes even outperform their neural counterparts on certain benchmarks \citep{steck2019embarrassingly,ferrari2019we,rendle2022revisiting}.

Implicit feedback presents a unique challenge: how should we handle the positive feedback (e.g. items clicked, watched) vs. the negative or missing feedback, typically represented as ones and zeros, respectively, in the user-item interaction matrix? Intuitively, we expect the items a user has interacted with at least once (positive feedback) to convey more information about that user's preference. If a user has never interacted with an item it could be due to a lack of interest, but it is more likely due to a lack of awareness/exposure \citep{liang2016modeling}.
Consequently, a common approach is to weight the data differently depending on whether their value is zero or one \citep{hu2008collaborative,pan2008}: instances of positive feedback (ones) are typically upweighted to reflect the prior belief that these observations are more informative. Previous studies have empirically demonstrated that such weighting strategies generally improve recommendation accuracy \citep{hu2008collaborative,pan2008,steck2010,liang2016modeling}. 

A notable exception to the success of weighting is the EASE model of \citet{steck2019embarrassingly}, which learns an item-item similarity matrix from unweighted data. Even though EASE is a full rank model, and the main focus of this paper is on low-rank matrix factorization, the fact that EASE can outperform not only weighted matrix factorization, but also powerful deep learning-based approaches, suggests there is more to be uncovered here.

In this paper we systematically study weighting schemes and matrix factorization algorithms\footnote{Our code is publicly available at \url{https://github.com/aa14k/Weighted-MF}} and find, somewhat surprisingly, that
the benefits of weighting are more nuanced and circumstantial than the conventional wisdom suggests.
Specifically, we observe that the best performing methods, as measured by (unweighted) ranking accuracy on publicly available datasets, are generally \textit{large models} trained using \textit{unweighted data}. On the other hand, there are cases where weighting can be beneficial, especially for \textit{small models} with lower capacity. These findings parallel recent observations about deep models \citep{byrd2019effect}.  
We highlight the following as the main contributions of this work:
\begin{enumerate}[leftmargin=*]
    \item We provide experimental evidence that unweighted large linear models often perform comparably to, or better than, their weighted counterparts on standard recommender systems benchmarks, contrary to the conventional wisdom in the literature. 
    \item We derive efficient algorithms to minimize weighted objectives for various matrix factorization models. To the best of our knowledge, most of these objectives were never studied in the literature because they could not be optimized effectively.
    \item We systematically study the interplay between weighting, regularization, and model capacity by running exhaustive experiments over a wide range of methods.
\end{enumerate}

\section{Preliminaries}

In this paper we consider the problem of top-N recommendation \citep{deshpande2004item,ning2011slim} with implicit feedback data \citep{hu2008collaborative,pan2008}. 
We assume access to a (typically sparse) matrix of user-item interactions $X \in \mathbb{R}^{|\mathcal{U}|\times|\mathcal{I}|}$, where $\mathcal{U} = \{1,2,\ldots,|\mathcal{U}|\}$ and $\mathcal{I} = \{1,2,\ldots,|\mathcal{I}|\}$ denote the sets of users and items, respectively. We assume that $|\mathcal{U}| \gg |\mathcal{I}|$, which is often the case in recommendation scenarios. If $X_{ui}$ is positive then we say user $u$ has interacted with item $i$, whereas $X_{ui} = 0$ corresponds to no interactions.

Following standard practice \citep{liang2018variational,steck2019embarrassingly,ferrari2019we}, we ``binarize'' $X$ by setting the positive entries to one while leaving the zero entries unchanged. This is done to facilitate the use of standard ranking metrics such as the truncated normalized discounted cumulative gain (nDCG$@R$) \citep{jarvelin2000ir} and Recall$@R$. In top-N recommendation we provide $N$ ranked recommendations to each user, and success is measured by the alignment between our recommendations and a withheld set of items the user actually interacted with (Recall$@R$), potentially considering the rankings (nDCG$@R$).

In the remainder of this section, we outline the major categories of models that we study, starting with the standard form of weighted matrix factorization.

\parhead{Weighted matrix factorization. } In weighted matrix factorization (WMF) \citep{hu2008collaborative,pan2008}, we learn $d$-dimensional user and item factors $U\in \mathbb{R}^{|\mathcal{U}|\times d}$ and $V\in \mathbb{R}^{|\mathcal{I}|\times d}$, with $d \le |\mathcal{I}|$, via the objective
\begin{equation}\label{eqn:wmf}
    \min_{U \in \mathbb{R}^{|\mathcal{U}|\times d}, V \in \mathbb{R}^{|\mathcal{I}|\times d}}
    \left\lVert \sqrt{W} \odot (X - UV^\top)\right\rVert_F^2 + \lambda \|U\|_F^2\ + \lambda \|V\|_F^2\,,
\end{equation}
where $\lambda \ge 0$ is the regularizer strength, $W \in \mathbb{R}_{>0}^{|\mathcal{U}|\times |\mathcal{I}|}$ is the (positive) weight matrix, $\sqrt{W}$ is computed element-wise, $\norm{\cdot}_F$ denotes the Frobenius norm, and $\odot$ denotes the Hadamard (element-wise) product. Typically $W$ is set such that the weights $W_{ui}$ corresponding to $X_{ui} = 1$ are larger than those corresponding to $X_{ui} = 0$. We follow \citet{hu2008collaborative}, \citet{pan2008}, and \citet{steck2010} by using the weight matrix $W_{ui} = (\alpha-1)X_{ui} + 1$ with $\alpha \ge 1$ for all of our experiments. Given this weighting, WMF can be efficiently optimized via alternating least squares \citep{hu2008collaborative}.

\parhead{Asymmetrical weighted matrix factorization.} 
Asymmetrical weighted matrix factorization (AWMF) \citep{paterek2007improving} is a closely related approach with the objective\footnote{The original A(W)MF of \citet{paterek2007improving} is unweighted.}
\begin{align}\label{eqn:awmf-weightdecay}
    \min_{U,V\in\mathbb{R}^{|\mathcal{I}|\times d}} \left\lVert \sqrt{W} \odot (X - XUV^\top)\right\rVert_F^2 + \lambda\|U\|_F^2 + \lambda \|V\|_F^2\,.
\end{align}
Note that in AWMF both the learned matrices $U$ and $V$ are of size $|\mathcal{I}| \times d$. 
Comparing the first terms in \cref{eqn:awmf-weightdecay,eqn:wmf}, we can interpret AWMF as having added a constraint on the latent user factors. Specifically, the factor for user $u$ in AWMF is implicitly defined as the sum of the item factors (in $U$) of items that $u$ has interacted with.
Compared to the unconstrained latent user factors of WMF, which must be learned for every user separately, the constrained formulation of AWMF results in a considerable reduction in the number of model parameters (the size of $U$ is reduced from $|\mathcal{U}|d$ to $|\mathcal{I}|d$). AWMF also has strong connections to autoencoder models \citep{steck2015gaussian,liang2018variational}.

Another variant of AWMF optimizes the objective
\begin{align}\label{eqn:awmf-dropout}
    \min_{U,V\in\mathbb{R}^{|\mathcal{I}|\times d}} \left\lVert \sqrt{W} \odot (X - XUV^\top)\right\rVert_F^2 + \lambda\|UV^\top\|_F^2\,.
\end{align}
The regularization in \cref{eqn:awmf-weightdecay} is ``weight decay'' style, while \cref{eqn:awmf-dropout} corresponds asymptotically to a ``dropout/denoising'' style regularization \citep{cavazza2018dropout}. 
The unweighted versions of \cref{eqn:awmf-weightdecay,eqn:awmf-dropout} are studied by \citet{steck_autoencoders}, where, without weighting, \cref{eqn:awmf-dropout} is demonstrated to perform better than \cref{eqn:awmf-weightdecay} by a fairly sizable margin.

Finally, we consider a full rank version of \cref{eqn:awmf-dropout}:
\begin{equation}\label{eqn:ease}
    \min_{B \in \mathbb{R}^{|\mathcal{I}| \times |\mathcal{I}|}} \left\lVert \sqrt{W} \odot (X - XB)\right\rVert_F^2 + \lambda \|B\|_F^2\,.
\end{equation}
This is similar to the objective of EASE \citep{steck2019embarrassingly}, except that we have introduced the weight matrix $W$ and dropped the constraint that $B$ has a zero-diagonal. The latter choice was made to be consistent with the other models for a fair comparison, even though this results in slightly reduced accuracy as shown by \citet{steck2019collaborative}. We expect the full rank model to generally outperform its low-rank counterparts. 

Several further comments are in order:
\begin{itemize}[leftmargin=*]
    \item We choose the square loss, $\| \cdot \|_F^2$, as it is standard in the literature when applying matrix factorization to collaborative filtering/top-N recommendations \citep{hu2008collaborative,ning2011slim,steck2019embarrassingly,steck_autoencoders}. \citet{liang2018variational} and \citet{depauw24} observe that training with alternative loss functions, such as the logistic loss, yields better ranking accuracy than the \textit{unweighted} square loss, as the logistic loss reweights the data \citep{ayoubswitching}. However, as our experiments demonstrate, the benefits of reweighting diminish with size for large linear models.
    \item We test various forms of $\ell_2$ regularization. As was observed by \citet{steck_autoencoders,steck2021regularization} and \citet{jin2021towards}, for an item similarity matrix with low rank, i.e. $B=UV^\top$, choosing the right regularization scheme can have a significant impact on ranking accuracy. This necessitates setting the hyperparameter $\lambda$---using a held-out validation set---separately for each regularizer. 
\end{itemize}
Note that we included the weight matrix explicitly in all of the objectives in this section, but the unweighted versions are recovered by setting $W$ to be all ones. For the objectives in \cref{eqn:awmf-weightdecay,eqn:awmf-dropout,eqn:ease} with a general weight matrix, there was previously no known algorithm capable of finding the closed-form solution\footnote{In the case of \cref{eqn:awmf-weightdecay,eqn:awmf-dropout}, this means the closed-form solution for either $U$ or $V$, given that the other is fixed.} on large-scale recommendation datasets. In the subsequent sections, we dive into the details of our algorithm design. 
\section{Optimization}

In this section, we derive closed-form solutions to the optimization problems given in \cref{eqn:awmf-weightdecay,eqn:awmf-dropout,eqn:ease}. In \cref{sec:alg}, these solutions will be paired with the conjugate gradient method to produce efficient algorithms.

\subsection{Closed-Form Solution: Unregularized}
We show that the objective functions for learning the item similarity matrices $UV^\top$ in \cref{eqn:awmf-dropout,eqn:awmf-weightdecay} and $B$ in \cref{eqn:ease}
 can be solved in closed form when $\lambda = 0$. We present the unregularized case first to simplify the exposition, but the proceeding subsection contains the extension to $\lambda > 0$. 
 
 Our results refute the claim made by \citet{steck2019collaborative} and \citet{jin2021towards} that a closed-form solution to \cref{eqn:ease,eqn:awmf-dropout,eqn:awmf-weightdecay} does not exist for an arbitrary weight matrix $W\in\mathbb{R}_{>  0}^{|\mathcal{U}|\times|\mathcal{I}|}$. Thus our closed-form solution to these weighted objectives resolves an open question.
The derivations require a review of the Kronecker product. 
\begin{definition}[Kronecker product]
    The Kronecker product of matrices $A \in \mathbb{R}^{m\times n}$ and $B \in \mathbb{R}^{p\times q}$ is denoted by $A \otimes B$ and defined to be the block matrix
    \begin{equation*}
        A\otimes B = 
        \begin{bmatrix}
        A_{11}B & A_{12}B & \cdots & A_{1n}B \\
        A_{21}B & A_{22}B & \cdots & A_{2n}B \\
        \vdots  & \vdots  & \ddots & \vdots  \\
        A_{m1}B & A_{m2}B & \cdots & A_{mn}B\
        \end{bmatrix} \in \mathbb{R}^{mp\times nq}\,.
    \end{equation*}
    We refer the reader to Chapter 4.2 of \citet{Horn_Johnson_1991} for a summary of elementary properties of the Kronecker product, some of which will be used below. We will also need the following ``vectorization'' lemma.
    \begin{lemma}[Lemma 4.3.1 of \citet{Horn_Johnson_1991}]\label{lem:vectorization}
        Let matrices $A \in \mathbb{R}^{m\times n},B \in \mathbb{R}^{p\times q}$ and $C \in \mathbb{R}^{m \times q}$ be given and $X \in \mathbb{R}^{n \times p}$ be unknown. The matrix equation
        \begin{equation*}
            AXB = C
        \end{equation*}
        is equivalent to the system of $qm$ equations in $np$ unknowns given by
        \begin{equation*}
            (B^\top \otimes A)\ve(X) = \ve(C)\,.
        \end{equation*}
        That is, $\ve(AXB) = (B^\top \otimes A)\ve(X)$, where $\ve(A)$ is the $mn\times 1$ column vector obtained by stacking the columns of $A$ with the leftmost column on top \footnote{$\ve(A)=A.$flatten("F") in NumPy.}.
    \end{lemma}
\end{definition}
The vectorization lemma has historically been used to reason about and compute solutions to Sylvester's equation of the form $AX + XB = C$ and the discrete-time Lyapunov equation of the form $AXA^\top - X + B = 0$, which arise naturally when studying linear dynamical systems in control theory. In this context the vectorization lemma ``linearizes'' these seemingly nonlinear matrix equations, and our closed-form solutions are inspired by this technique.

We will employ the following proposition repeatedly in our derivations to avoid repetitive calculations. It uses both $I_m$ to denote the identity matrix of size $m\times m$ and $\di(x)$ to denote the matrix with $x \in \mathbb{R}^m$ on the diagonal and zeros elsewhere.
\begin{proposition}\label{prop:vec-model}
    For matrices $A,W \in \mathbb{R}^{u\times i}$ and $B,C \in \mathbb{R}^{i\times d}$, 
    \begin{equation}\label{prop-eq1}
        \ve(A^\top(W\odot(ABC^\top))C) = (C\otimes A)^\top \bar{W}(C\otimes A)\ve(B)
    \end{equation}
    where $\bar{W} = \di(\ve(W))$. Furthermore, we have
    \begin{equation}\label{prop-eq2}
        \ve((W^\top\odot(CB^\top A^\top))AB) = (AB\otimes I_i)^\top \hat{W}(AB\otimes I_i)\ve(C)\,,
    \end{equation}
    where $\hat W = \di(\ve(W^\top))$. 
\end{proposition}
\begin{proof}
    We begin with \cref{prop-eq1}, and applying \cref{lem:vectorization} produces
    \begin{align*}
        \MoveEqLeft \ve(A^\top(W\odot(ABC^\top))C)\\
        &= (C\otimes A)^\top \ve(W\odot(ABC^\top)) \\
        &= (C\otimes A)^\top (\ve(W)\odot\ve(ABC^\top))\\
        &= (C\otimes A)^\top \di(\ve(W)) \ve(ABC^\top) \\
        &= (C\otimes A)^\top \bar{W} (C\otimes A)\ve(B)\,,
    \end{align*}
    where the first and last equalities use \cref{lem:vectorization} and the third equality uses the fact that $x\odot y = \di(x)y$ for vectors $x,y$. 

    Now we turn our attention to showing \cref{prop-eq2}. Again employing \cref{lem:vectorization}, we have that
    \begin{align*}
        \MoveEqLeft \ve((W^\top\odot(CB^\top A^\top))AB) \\
        &= (AB\otimes I_i)^\top \ve(W^\top\odot(C(AB)^\top)) \\
        &= (AB\otimes I_i)^\top (\ve(W^\top)\odot\ve(C(AB)^\top)) \\
        &= (AB\otimes I_i)^\top \di(\ve(W^\top))\ve(C(AB)^\top) \\
        &= (AB\otimes I_i)^\top \hat{W}(AB\otimes I_i)\ve(C)\,,
    \end{align*}
    where the first and last equalities use \cref{lem:vectorization} and the third equality uses again that $x\odot y = \di(x)y$ for vectors $x,y$.
\end{proof}

With these preliminary results in hand, we are ready to tackle \cref{eqn:ease,eqn:awmf-dropout,eqn:awmf-weightdecay} with $\lambda = 0$. 
\paragraph{\textbf{Solving \cref{eqn:ease}, unregularized:}} We start by solving the optimization problem given in \cref{eqn:ease}, as it is the simplest to solve and thus provides the cleanest presentation of the analysis. Setting the derivative of \cref{eqn:ease} to zero produces
\begin{equation*}
     X^\top\left(W \odot(XB - X) \right) = 0\,.
\end{equation*}
Previous attempts at finding a closed form solution failed due to the presence of the Hadamard (element-wise) product $\odot$, which must be evaluated before the multiplication by $X^\top$ (i.e. $\odot$ is not associative with matrix multiplication). 
We can rearrange the above matrix equation and then apply \cref{prop-eq1} of \cref{prop:vec-model} for
\begin{align*}\label{eqn:ease-calcuations}
    \ve(X^\top(W\odot X))
    &= \ve\left(X^\top\left(W \odot(XB)\right)\right)\\
    &= (I_{|\mathcal{I}|}\otimes X^\top)\bar W(I_{|\mathcal{I}|}\otimes X)\ve(B)\\
    &= H(B,0)\ve(B)\,,
\end{align*}
where we defined\footnote{We slightly abuse notation here by writing $H(B,0)$, since $H$ is not a function of $B$ but is a function of $0$. Later on, 0 will be replaced by $\lambda$.} $H(B,0) = (I_{|\mathcal{I}|}\otimes X^\top)\bar W(I_{|\mathcal{I}|}\otimes X)$.
Rearranging the above equation gives the closed form for $B$,
\begin{equation}\label{eqn:closed-form-ease}
    \ve(B) = H(B,0)^{-1} \ve(X^\top(W\odot X))\,.
\end{equation}
We remark that $H(B,0)$ is guaranteed to be invertible so long as $X$ has full column rank, i.e. $X^\top X$ is invertible.  
Finally, we note that, while a closed form solution for $B$ was given in terms of $W$ and $X$, this method can be used to find a closed form solution for any optimization problem of the form
\begin{equation}\label{eqn:generic-wease}
    \min_{B\in\mathbb{R}^{|\mathcal{I}|\times |\mathcal{I}|}} \left\lVert \sqrt{W} \odot (X' - XB)\right\rVert_F^2\,,
\end{equation}
where $X' \in \mathbb{R}^{|\mathcal{U}|\times |\mathcal{I}|}$, which may be of independent interest. An interesting extension would be to leverage this analysis and derive a closed form solution for the case when $B$ is constrained to be a zero-diagonal matrix; this constraint was enforced by \citet{steck2019embarrassingly} and \citet{ning2011slim}. We believe the use of the vectorization lemma will facilitate such an analysis and we highlight this as potential future work.  
\paragraph{\textbf{Solving \cref{eqn:awmf-dropout,eqn:awmf-weightdecay}, unregularized:}} We now apply \cref{prop:vec-model} to \cref{eqn:awmf-dropout,eqn:awmf-weightdecay} with 
$\lambda = 0$ using a similar procedure. Since we will employ alternating minimization for solving \cref{eqn:awmf-dropout,eqn:awmf-weightdecay}, we first fix $V$ and compute the closed form solution for $U$ and then do the converse. Taking the derivative of \cref{eqn:awmf-dropout} with respect to $U$ and setting it to zero gives
\begin{equation*}
    X^\top(W\odot(XUV^\top - X))V = 0\,.
\end{equation*}
Employing \cref{prop-eq1} of \cref{prop:vec-model} and rearranging the above equation gives
\begin{align*}
    \MoveEqLeft\ve(X^\top(W\odot X)V)\\
    &=  \ve(X^\top(W\odot(XUV^\top))V) \\
    &=  (V\otimes X)^\top\bar{W}(V\otimes X)\ve(U)\\
    &= H_D(U,0)\ve(U) = H_W(U,0)\ve(U) \,,
\end{align*}
where $H_D(U,0) = H_W(U,0) = (V\otimes X)^\top\bar{W}(V\otimes X)$. In the regularized case,  $H_D(U,\lambda)$ will denote the Gram matrix of AWMF with respect to $U$ with dropout, while $H_W(U,\lambda)$ will correspond to weight decay. These quantities agree when $\lambda = 0$, but
in later sections when we consider $\lambda > 0$ the utility of these shorthands will become more apparent.
Rearranging the above equation gives the closed-form expression for $U$,
\begin{align}\label{eqn:awmfU-soln}
    \ve(U)
    = \big((V\otimes X)^\top\bar{W}(V\otimes X)\big)^{-1} \ve(X^\top(W\odot X)V)\,.
\end{align}
So long as both $V$ and $X$ have full column rank, the inverse above is guaranteed to exist.
Repeating the same calculations for $V$ and using \cref{prop-eq2} of \cref{prop:vec-model}, we get that
\begin{equation}\label{eqn:awmfV-soln}
    \ve(V) = \big((XU \otimes I_{|\mathcal{I}|})^\top\hat{W}(XU\otimes I_{|\mathcal{I}|})\big)^{-1} \ve((W^\top\odot X^\top)XU)\,.
\end{equation}
Here $XU$ must have full column rank for the inverse to exist. We will set
$H_D(V,0)=H_W(V,0) 
= (XU \otimes I_{|\mathcal{I}|})^\top\hat{W}(XU\otimes I_{|\mathcal{I}|})$, similarly to before. This concludes our analysis of the case $\lambda = 0$, and we are ready for the generalization to regularization.

\subsection{Closed-Form Solution: Regularized}
In this section, we show that the regularized objective functions for learning $B$, $U$ and $V$ in \cref{eqn:ease,eqn:awmf-dropout,eqn:awmf-weightdecay} can be solved in closed form. \paragraph{\textbf{Solving \cref{eqn:ease}, regularized:}} Taking the derivative of the optimization problem in \cref{eqn:ease} and setting it to be zero, we get
\begin{equation*}
    X^\top\left(W \odot(XB - X) \right) + \lambda B = 0\,.
\end{equation*}
Rearranging the above equation and employing \cref{prop:vec-model}, we write
\begin{align*}
    \MoveEqLeft\ve(X^\top(W\odot X))\\
    &= \ve\left(X^\top\left(W \odot(XB) \right) + \lambda B\right) \\
    &= H(B,0)\ve(B) + \lambda I_{|\mathcal{I}|^2}\ve(B) \\
    &= H(B,\lambda)\ve(B)\,,
\end{align*}
where we defined $H(B,\lambda) =  H(B,0) + \lambda I_{|\mathcal{I}|^2}$.
Thus the closed form expression for $B$ when solving \cref{eqn:ease} with $\lambda > 0$ is
\begin{equation}\label{eqn:cf-ease-reg}
    B = H(B,\lambda)^{-1} \ve(X^\top(W\odot X))\,.
\end{equation}
\paragraph{\textbf{Solving \cref{eqn:awmf-dropout,eqn:awmf-weightdecay}, regularized:}} We now turn our attention to the optimization problems given in \cref{eqn:awmf-dropout,eqn:awmf-weightdecay}, and begin by fixing $V$ and finding the closed form expressions for $U$. Setting the derivative of \cref{eqn:awmf-dropout} with respect to $U$ to zero gives
\begin{equation*}
    X^\top(W\odot(XUV^\top - X))V + \lambda  UV^\top V = 0 \,.
\end{equation*}
Again, rearranging and employing \cref{prop:vec-model} gives
\begin{align*}
    \MoveEqLeft \ve(X^\top(W\odot X)V)\\
    &= \ve(X^\top(W\odot(XUV^\top))V + \lambda UV^\top V)\\
    &= (H_D(U,0) + \lambda (V^\top V\otimes I_{|\mathcal{I}|}))\ve(U) \\
    &= H_D(U,\lambda)\ve(U)\,.
\end{align*}
Thus, when employing dropout style regularization, as in \cref{eqn:awmf-dropout}, we have that 
\begin{equation}\label{eqn:cfU-dropout}
    U_{\text{dropout}} =H_D(U,\lambda)^{-1} \ve(X^\top(W\odot X)V)\,.
\end{equation}
Now taking the derivative of \cref{eqn:awmf-weightdecay} with respect to $U$ and repeating the above calculations gives us the closed-form solution for $U$ when employing a weight decay style of regularization, 
\begin{equation}\label{eqn:cfU-weight-decay}
    U_{\text{weight decay}} = H_W(U,\lambda)^{-1} \ve(X^\top(W\odot X)V)\,,
\end{equation}
where $H_W(U,\lambda) = H_W(U,0) + \lambda I_{d|\mathcal{I}|}$.
Repeating the above calculations but solving for $V$ instead of $U$, \cref{eqn:awmf-dropout} gives 
\begin{equation}\label{eqn:cfV-dropout}
    V_{\text{dropout}} = H_D(V,\lambda)^{-1} \ve((W^\top\odot X^\top)XU)\,,
\end{equation}
where $H_D(V,\lambda) = H_D(V, 0) + \lambda(U^\top U \otimes I_{|\mathcal{I}|})$, and \cref{eqn:awmf-weightdecay} gives
\begin{equation}\label{eqn:cfV-weight-decay}
    V_{\text{weight decay}} = H_W(V,\lambda)^{-1} \ve((W^\top\odot X^\top)XU)\,,
\end{equation}
where $H_W(V,\lambda) = H_W(V,0) + \lambda I_{d|\mathcal{I}|}$. 
With the closed-form solutions for the minimization problems defined in \cref{eqn:ease,eqn:awmf-dropout,eqn:awmf-weightdecay} in place, we now turn our attention to computation.
\section{Computational Considerations}\label{sec:alg}

\begin{algorithm}[t]
\caption{Preconditioned Conjugate Gradient Method}\label{alg:conjugate-gradient}
\KwIn{$A$: symmetric positive definite matrix, $b$: vector, $x_0$: initial guess, $M^{-1}$: preconditioner, $\epsilon$: tolerance, $max\_iter$: maximum iterations}
\KwOut{$x$: solution to $Ax = b$}

$r_0 \leftarrow b - Ax_0$\\
$z_0 \leftarrow M^{-1}r_0$\\
$p_0 \leftarrow z_0$\\

\For{$k \leftarrow 0$ \KwTo $max\_iter - 1$}{
    \If{$\|r_k\|_2 \leq \epsilon$}{
        \Return{$x_k$}
    }
    $\alpha_k \leftarrow (r_k^\top z_k)/(p_k^\top A p_k)$\\
    $x_{k+1} \leftarrow x_k + \alpha_k p_k$\\
    $r_{k+1} \leftarrow r_k - \alpha_k A p_k$\\
    $z_{k+1} \leftarrow M^{-1}r_{k+1}$\\
    $\beta_k \leftarrow (r_{k+1}^\top z_{k+1})/(r_k^\top z_k)$\\
    $p_{k+1} \leftarrow z_{k+1} + \beta_k p_k$\\
}
\Return{$x_k$}
\end{algorithm}

In this section, we address some of the computational issues arising from the closed-form solutions given in the previous section, and detail efficient algorithms for computing these solutions. The main difficulty is that the memory requirements for computing the Gram matrices $H$, $H_D$ and $H_W$ can be prohibitively large. For example, suppose one wants to compute $H_D(U,\lambda)$ on the MovieLens $20$ Million (ML-
$20$M) dataset \citep{ml-20m}, which has $136,677$ users and $20,108$ items (movies) with around $10$ million interactions. A quick calculation gives us that the sparse ``binarized'' click matrix for ML-$20$M, i.e. $X$, has density $0.003$, meaning about $0.3$\% of the matrix $X$ is nonzero. If we let $d=10$ in \cref{eqn:awmf-dropout,eqn:awmf-weightdecay} and let $V \in \mathbb{R}^{|\mathcal{I}|\times d}$ be a dense matrix, then storing $V\otimes X$, which is required to compute $H_U$, would require more than $1,000$GB of memory.

We now give an efficient algorithm for solving \cref{eqn:cf-ease-reg,eqn:cfU-dropout,eqn:cfV-dropout,eqn:cfU-weight-decay,eqn:cfV-weight-decay}. The design is based on the following two key insights:
\begin{itemize} [leftmargin=*]
    \item The regularized Gram matrices ($H, H_D$ and $H_W$) corresponding to the aforementioned equations are symmetric positive definite.
    \item While computing and storing these matrices is prohibitively expensive, computing and storing their product with a vector is (relatively) inexpensive, and indeed tractable for problems of the scale considered in our experiments.
\end{itemize}
 In order to understand the second insight more concretely, define a matrix $P \in \mathbb{R}^{|\mathcal{I}|\times|\mathcal{I}|}$ and let $p = \ve(P)$. Then \cref{prop-eq1} of \cref{prop:vec-model} reveals the equivalence between computing $H(B,\lambda)p$ and
\begin{equation*}
    \ve\left(X^\top (W\odot(XP)) + \lambda P\right)\,.
\end{equation*}
If instead we have $P\in \mathbb{R}^{|\mathcal{I}|\times d}$ then $H_D(U,\lambda)p$, $H_D(V,\lambda)p$, $H_W(U,\lambda)p$ and $H_W(V,\lambda)p$ can be computed efficiently by invoking \cref{prop-eq1,prop-eq2} of \cref{prop:vec-model} in an analogous way.

Therefore, since the Gram matrices are symmetric positive definite and their products with arbitrary vectors are inexpensive to compute, the natural algorithm for solving \cref{eqn:cf-ease-reg,eqn:cfU-dropout,eqn:cfV-dropout,eqn:cfU-weight-decay,eqn:cfV-weight-decay} is the conjugate gradient method of \citet{cgmethod} with preconditioning \citep{axelsson1972generalized}, which is detailed in \cref{alg:conjugate-gradient}. The conjugate gradient method is an iterative method for solving positive definite linear systems of equations. The main computational bottleneck in \cref{alg:conjugate-gradient} is computing matrix-vector products but, as we argued above, this can be done (relatively) cheaply. Another pleasing property of \cref{alg:conjugate-gradient} is that it converges to the solutions of \cref{eqn:cf-ease-reg,eqn:cfU-dropout,eqn:cfV-dropout,eqn:cfU-weight-decay,eqn:cfV-weight-decay} in \textit{at most} $|\mathcal{I}|d$ iterations, as stated in Theorem 5.1 of \citet{nocedal1999numerical}. 

However, \cref{alg:conjugate-gradient} may identify a solution in fewer than $|\mathcal{I}|d$ iterations when the regularized Gram matrix ($A$ in the algorithm) is well behaved. Specifically, let the conditioning number of $A$ be defined by $\kappa(A) = \lambda_{\max}(A)/\lambda_{\min}(A)$, where $\lambda_{\max}(A)$ and $\lambda_{\min}(A)$ are the maximum and minimum eigenvalues of $A$, respectively. Then the values of $x_k$ for $k = 0,\ldots,max\_iter-1$, as computed by \cref{alg:conjugate-gradient}, satisfy
\begin{equation*}
    \|x_k - x\|_A \leq 2\exp \left(-\frac{2k}{\kappa(A)}\right)\|x_0-x\|_A\,.
\end{equation*}
 Thus \cref{alg:conjugate-gradient} converges exponentially fast to the solution of any positive definite linear system with a reasonably small conditioning number. 

In order to control the conditioning number of the various Gram matrices, we employ the preconditioner $M^{-1}$ to speed up the convergence of \cref{alg:conjugate-gradient} without increasing the computational complexity. At a high level, preconditioning transforms the linear system $Ax=b$ to the linear system $M^{-1}Ax=M^{-1}b$ for an invertible matrix $M$. The goal is to choose $M$ so that $M^{-1}A$ has a lower conditioning number than $A$, and thus the conjugate gradient method has a faster convergence rate on the transformed system. 
Empirically we found that the Gram matrix corresponding to $W=1$ was performant, e.g. for $H(B,\lambda)$ we would use $M^{-1}=((\Ii\otimes X^\top X) + \lambda I_{|\mathcal{I}|^2})^{-1} = \Ii\otimes (X^\top X+\lambda \Ii)^{-1}$. Intuitively, this works because the conditioning number of $M^{-1}H$ is of order $\max(W)$, which in our experiments is controlled by $\alpha$. The preconditioner $M^{-1}$ corresponding to \cref{eqn:cf-ease-reg,eqn:cfU-dropout,eqn:cfV-dropout,eqn:cfU-weight-decay,eqn:cfV-weight-decay} can be computed efficiently via \cref{prop:vec-model}.
\section{Numerical Experiments}\label{sec:numerical-experiments}
\renewcommand{\arraystretch}{1.5}

\begin{table*}
\caption{Ranking accuracies (Recall$@20$, Recall$@50$ and nDCG$@100$, with standard errors of 0.002, 0.001, 0.001 for ML-20M, Netflix, and MSD, respectively) for linear models of rank $d=1000$ trained with different regularizers and weightings $\alpha$. For each model we report results for both the best-performing weighted model (and its weight $\alpha\neq1$) and the best-performing unweighted model ($\alpha=1$). In each dataset (ML-20M, Netflix, MSD), the higher value for a given metric is boldfaced; in case of a tie both entries are boldfaced (and that metric does not contribute to the win count). The “Wins” column (formatted like the dataset groups) reports, out of 9 metrics, how many times each method wins.}
\label{tab-d1000}
\resizebox{\textwidth}{!}{
\begin{tabular}{lrrrrcrrrcrrrcc}
 & & \multicolumn{4}{c}{{\bf ML-20M}} & \multicolumn{4}{c}{{\bf Netflix}} & \multicolumn{4}{c}{{\bf MSD}} & \multicolumn{1}{c}{{\bf Wins}} \\
\cmidrule(lr){3-6} \cmidrule(lr){7-10} \cmidrule(lr){11-14} \cmidrule(lr){15-15}
 & & Recall & Recall & nDCG & \multirow{2}{*}{$\alpha$} & Recall & Recall & nDCG & \multirow{2}{*}{$\alpha$} & Recall & Recall & nDCG & \multirow{2}{*}{$\alpha$} &  \\
\multicolumn{2}{l}{{\bf Model}} & @20 \hfill{} & @50 \hfill{} & @100\hfill{} & & @20 \hfill{} & @50 \hfill{} & @100 \hfill{} & & @20 \hfill{}   & @50 \hfill{} & @100 &  & \\
\cmidrule(lr){1-2} \cmidrule(lr){3-6} \cmidrule(lr){7-10} \cmidrule(lr){11-14} \cmidrule(lr){15-15}
1. & $\n{\sqrt W\odot (X - XUV\t)} + \lambda(\n{U} + \n{V})$
    & 0.344          
    & \textbf{0.476} 
    & 0.378          
    & 2 
    & \textbf{0.311} 
    & \textbf{0.395} 
    & \textbf{0.345} 
    & 5 
    & \textbf{0.228} 
    & \textbf{0.314} 
    & \textbf{0.277} 
    & 10
    & 7/9 \\
   & $\n{X - XUV\t} + \lambda(\n{U} + \n{V})$
    & \textbf{0.348} 
    & 0.466         
    & \textbf{0.381} 
    & 1 
    & 0.302         
    & 0.376         
    & 0.334         
    & 1 
    & 0.189         
    & 0.251         
    & 0.230         
    & 1
    & 2/9 \\[2mm]
2. & $\n{\sqrt W\odot (X - UV\t)} + \lambda(\n{U} + \n{V})$
    & \textbf{0.377} 
    & \textbf{0.512} 
    & \textbf{0.409} 
    & 3
    & 0.292         
    & 0.376         
    & 0.329         
    & 2
    & \textbf{0.268} 
    & \textbf{0.371} 
    & \textbf{0.322} 
    & 21
    & 6/9 \\
   & $\n{X - UV\t} + \lambda(\n{U} + \n{V})$
    & 0.328         
    & 0.443         
    & 0.360         
    & 1
    & \textbf{0.318} 
    & \textbf{0.406} 
    & \textbf{0.352} 
    & 1
    & 0.189         
    & 0.251         
    & 0.231         
    & 1
    & 3/9 \\[2mm]
3. & $\n{\sqrt W\odot (X - XUV\t)} + \lambda\n{UV\t}$
    & 0.355         
    & 0.491         
    & 0.388         
    & 2
    & \textbf{0.337} 
    & \textbf{0.419} 
    & \textbf{0.370} 
    & 2
    & \textbf{0.227} 
    & \textbf{0.311} 
    & \textbf{0.277} 
    & 10
    & 4/9 \\
   & $\n{X - XUV\t} + \lambda\n{UV\t}$
    & \textbf{0.378} 
    & \textbf{0.511} 
    & \textbf{0.407} 
    & 1
    & \textbf{0.337} 
    & 0.417         
    & 0.369         
    & 1
    & 0.216         
    & 0.296         
    & 0.266         
    & 1
    & 5/9 \\[2mm]
4. & $\n{\sqrt W\odot (X - XUV\t)} + \lambda(\n{XU} + \n{V})$
    & \textbf{0.356} 
    & \textbf{0.491} 
    & \textbf{0.383} 
    & 2
    & 0.263         
    & 0.355         
    & 0.303         
    & 2
    & \textbf{0.231} 
    & \textbf{0.322} 
    & \textbf{0.283} 
    & 10
    & 6/9 \\
   & $\n{X - XUV\t} + \lambda(\n{XU} + \n{V})$
    & 0.326         
    & 0.441         
    & 0.358         
    & 1
    & \textbf{0.314} 
    & \textbf{0.401} 
    & \textbf{0.348} 
    & 1
    & 0.190         
    & 0.251         
    & 0.231         
    & 1
    & 3/9 \\[2mm]
5.  &\text{Full Rank} \quad $\n{\sqrt W\odot (X - XB)} + \lambda(\n{B})$
    & 0.360         
    & 0.499         
    & 0.389         
    & 2
    & \textbf{0.338} 
    & \textbf{0.421} 
    & \textbf{0.370} 
    & 2
    & 0.277         
    & 0.377         
    & 0.338         
    & 2
    & 3/9 \\
   & $\n{X - XB} + \lambda(\n{B})$
    & \textbf{0.376} 
    & \textbf{0.511} 
    & \textbf{0.407} 
    & 1
    & 0.335         
    & 0.417         
    & 0.368         
    & 1
    & \textbf{0.284} 
    & \textbf{0.384} 
    & \textbf{0.344} 
    & 1
    & 6/9 \\
\cmidrule(lr){1-15}
\end{tabular}}
\end{table*}

\begin{table*}
\caption{Ranking accuracies for linear models of rank $d=100$. }
\label{tab-d100}
\resizebox{\textwidth}{!}{
\begin{tabular}{lrrrrcrrrcrrrcc}
 & & \multicolumn{4}{c}{{\bf ML-20M}} & \multicolumn{4}{c}{{\bf Netflix}} & \multicolumn{4}{c}{{\bf MSD}} & \multicolumn{1}{c}{{\bf Wins}} \\
\cmidrule(lr){3-6} \cmidrule(lr){7-10} \cmidrule(lr){11-14} \cmidrule(lr){15-15}
 & & Recall & Recall & nDCG & \multirow{2}{*}{$\alpha$} & Recall & Recall & nDCG & \multirow{2}{*}{$\alpha$} & Recall & Recall & nDCG & \multirow{2}{*}{$\alpha$} &  \\
\multicolumn{2}{l}{{\bf Model}} & @20 \hfill{} & @50 \hfill{} & @100\hfill{} & & @20 \hfill{} & @50 \hfill{} & @100 \hfill{} & & @20 \hfill{}   & @50 \hfill{} & @100 &  & \\
\cmidrule(lr){1-2} \cmidrule(lr){3-6} \cmidrule(lr){7-10} \cmidrule(lr){11-14} \cmidrule(lr){15-15}
1. & $\n{\sqrt W\odot (X - XUV\t)} + \lambda(\n{U} + \n{V})$ 
    & 0.339           
    & \textbf{0.472}  
    & 0.374           
    & 2 
    & \textbf{0.309}  
    & \textbf{0.392}  
    & \textbf{0.343}  
    & 5 
    & \textbf{0.167}  
    & \textbf{0.241}  
    & \textbf{0.210}  
    & 10 
    & 7/9 \\
   & $\n{X - XUV\t} + \lambda(\n{U} + \n{V})$          
    & \textbf{0.348}  
    & 0.466           
    & \textbf{0.381}  
    & 1 
    & 0.302 
    & 0.376 
    & 0.334 
    & 1 
    & 0.129 
    & 0.182 
    & 0.163 
    & 1 
    & 2/9 \\[2mm]
2. & $\n{\sqrt W\odot (X - UV\t)} + \lambda(\n{U} + \n{V})$  
    & \textbf{0.364} & \textbf{0.497} & \textbf{0.399} & 6 
    & \textbf{0.316} & \textbf{0.396} & \textbf{0.348} & 3 
    & \textbf{0.178} & \textbf{0.265} & \textbf{0.223} & 21 
    & 9/9 \\
   & $\n{X - UV\t} + \lambda(\n{U} + \n{V})$           
    & 0.319 & 0.434 & 0.350 & 1 
    & 0.312 & 0.388 & 0.344 & 1 
    & 0.126 & 0.179 & 0.160 & 1 
    & 0/9 \\[2mm]
3. & $\n{\sqrt W\odot (X - XUV\t)} + \lambda\n{UV\t}$        
    & 0.320  & 0.454  & 0.361  & 2 
    & 0.316  & \textbf{0.396}  & \textbf{0.348}  & 2   
    & \textbf{0.167}  & \textbf{0.239}  & \textbf{0.208}  & 10 
    & 4/9 \\
   & $\n{X - XUV\t} + \lambda\n{UV\t}$                 
    & \textbf{0.336}  & \textbf{0.466}  & \textbf{0.368}  & 1 
    & \textbf{0.318}  & 0.395  & \textbf{0.348}  & 1   
    & 0.132  & 0.186  & 0.167  & 1 
    & 4/9 \\[2mm]
4. & $\n{\sqrt W\odot (X - XUV\t)} + \lambda(\n{XU} + \n{V})$
    & \textbf{0.344}  & \textbf{0.475}  & \textbf{0.374}  & 2 
    & 0.304  & 0.384  & 0.338  & 2 
    & \textbf{0.171}  & \textbf{0.247}  & \textbf{0.213}  & 20 
    & 6/9 \\
   & $\n{X - XUV\t} + \lambda(\n{XU} + \n{V})$         
    & 0.315  & 0.430  & 0.346  & 1 
    & \textbf{0.312}  & \textbf{0.388}  & \textbf{0.343}  & 1 
    & 0.128  & 0.180  & 0.161  & 1 
    & 3/9 \\
\cmidrule(lr){1-15}
\end{tabular}}
\end{table*}

\begin{table*}
\caption{Ranking accuracies for linear models of rank $d=10$.}
\label{tab-d10}
\centering
\resizebox{\textwidth}{!}{
\begin{tabular}{lrrrrcrrrcrrrcc}
 & & \multicolumn{4}{c}{{\bf ML-20M}} & \multicolumn{4}{c}{{\bf Netflix}} & \multicolumn{4}{c}{{\bf MSD}} & \multicolumn{1}{c}{{\bf Wins}} \\
\cmidrule(lr){3-6} \cmidrule(lr){7-10} \cmidrule(lr){11-14} \cmidrule(lr){15-15}
 & & Recall & Recall & nDCG & \multirow{2}{*}{$\alpha$} & Recall & Recall & nDCG & \multirow{2}{*}{$\alpha$} & Recall & Recall & nDCG & \multirow{2}{*}{$\alpha$} &  \\
\multicolumn{2}{l}{{\bf Model}} & @20 \hfill{} & @50 \hfill{} & @100\hfill{} & & @20 \hfill{} & @50 \hfill{} & @100 \hfill{} & & @20 \hfill{}   & @50 \hfill{} & @100 &  & \\
\cmidrule(lr){1-2} \cmidrule(lr){3-6} \cmidrule(lr){7-10} \cmidrule(lr){11-14} \cmidrule(lr){15-15}
1. & $\n{\sqrt W\odot (X - XUV\t)} + \lambda(\n{U} + \n{V})$
    & \textbf{0.305}  
    & \textbf{0.426}  
    & \textbf{0.338}  
    & 5 
    & \textbf{0.259}  
    & \textbf{0.339}  
    & \textbf{0.293}  
    & 2 
    & \textbf{0.086}  
    & \textbf{0.134}  
    & \textbf{0.114}  
    & 20
    & 9/9 \\
   & $\n{X - XUV\t} + \lambda(\n{U} + \n{V})$
    & 0.297         
    & 0.415         
    & 0.332         
    & 1 
    & 0.258         
    & 0.335         
    & 0.290         
    & 1 
    & 0.074         
    & 0.111         
    & 0.099         
    & 1
    & 0/9 \\[2mm]
2. & $\n{\sqrt W\odot (X - UV\t)} + \lambda(\n{U} + \n{V})$
    & \textbf{0.305}  
    & \textbf{0.427}  
    & \textbf{0.339}  
    & 3 
    & 0.258         
    & \textbf{0.340}  
    & \textbf{0.293}  
    & 3 
    & \textbf{0.082}  
    & \textbf{0.130}  
    & \textbf{0.111}  
    & 21
    & 8/9 \\
   & $\n{X - UV\t} + \lambda(\n{U} + \n{V})$
    & 0.289         
    & 0.401         
    & 0.322         
    & 1 
    & \textbf{0.259}  
    & 0.337         
    & 0.292         
    & 1 
    & 0.074         
    & 0.110         
    & 0.099         
    & 1
    & 1/9 \\[2mm]
3. & $\n{\sqrt W\odot (X - XUV\t)} + \lambda\n{UV\t}$
    & \textbf{0.303}  
    & \textbf{0.425}  
    & \textbf{0.336}  
    & 5 
    & \textbf{0.260}  
    & \textbf{0.340}  
    & \textbf{0.293}  
    & 2 
    & \textbf{0.085}  
    & \textbf{0.129}  
    & \textbf{0.111}  
    & 10
    & 9/9 \\
   & $\n{X - XUV\t} + \lambda\n{UV\t}$
    & 0.298         
    & 0.415         
    & 0.331         
    & 1 
    & 0.258         
    & 0.334         
    & 0.290         
    & 1 
    & 0.075         
    & 0.110         
    & 0.098         
    & 1
    & 0/9 \\[2mm]
4. & $\n{\sqrt W\odot (X - XUV\t)} + \lambda(\n{XU} + \n{V})$
    & \textbf{0.293}  
    & \textbf{0.412}  
    & \textbf{0.324}  
    & 2 
    & \textbf{0.260}         
    & \textbf{0.339}  
    & \textbf{0.294}  
    & 2 
    & \textbf{0.083}  
    & \textbf{0.129}  
    & \textbf{0.110}  
    & 20
    & 8/9 \\
   & $\n{X - XUV\t} + \lambda(\n{XU} + \n{V})$
    & 0.290         
    & 0.407         
    & 0.321         
    & 1 
    & \textbf{0.260}  
    & 0.338         
    & 0.293         
    & 1 
    & 0.074         
    & 0.110         
    & 0.098         
    & 1
    & 0/9 \\
\cmidrule(lr){1-15}
\end{tabular}}
\end{table*}

In this section, we apply the models derived in Section 3, as well as the weighted matrix factorization model \citep{hu2008collaborative,pan2008}, to collaborative filtering and recommendation problems. We focus on weighted linear autoencoders to better understand the interplay between \textit{weighting} ($W$), \textit{regularization} ($\lambda$) and \textit{model-dimension} ($d$). 

\subsection{Experimental Setup}

We follow the experimental setup detailed by \citet{liang2018variational}, using their publicly available code, and employing the same preprocessing protocol; the user-by-item matrix $X$ is ``binarized'', i.e. set to $1$ where a user has interacted with an item (e.g., the user watched a movie and gave a rating above 3) and $0$ where a user has not interacted with an item.
We also experiment with the same three standard datasets as \citet{liang2018variational}:
\begin{itemize} [leftmargin=*]
    \item MovieLens $20$ Million (ML-20M) \citep{ml-20m}, which consists of $136,677$ users, $20,108$ items, and about $10$ million interactions after preprocessing;
    \item Netflix Prize (Netflix) \citep{bennett2007netflix}, which consists of $463,435$ users, $17,769$ items, and about $57$ million interactions after preprocessing; and
    \item Million Song Data (MSD) \citep{msd}, which consists of $571,355$ users, $41,140$ items, and about $34$ million interactions after preprocessing.
\end{itemize}

 We compare the following five models:
\begin{enumerate}[leftmargin=*]
    \item Asymmetrical Weighted Matrix Factorization (AWMF) with weight decay \citep{paterek2007improving}. This corresponds to the optimization problem detailed in \cref{eqn:awmf-weightdecay}. The regularization here encourages the entries of $U$ and $V$ to be close to zero.
    \item Weighted Matrix Factorization (WMF) \citep{hu2008collaborative,pan2008}. A linear model that learns a user embedding matrix $U$ and an item embedding matrix $V$ that minimize $\|\sqrt{W}\odot(X-UV^\top)\| + \lambda(\|U\|_F^2 + \|V\|_F^2)$. This method also employs a weight decay style of regularization, encouraging the entries of $U$ and $V$ to be near zero. 
    \item Asymmetrical Weighted Matrix Factorization (AWMF) with dropout \citep{cavazza2018dropout,steck_autoencoders}. This corresponds to the optimization problem detailed in \cref{eqn:awmf-dropout}. The regularization here encourages $U$ to be orthogonal to $V^\top$, that is $UV^\top \approx 0$. This style of regularization was inspired by the regularizer originally employed by the EASE model of \citet{steck2019embarrassingly}.
    \item Asymmetrical Weighted Matrix Factorization (AWMF) with ``data/weight decay''. This model combines the regularization used by WMF with the training objective used by AWMF, i.e. $\min_{U,V}\|\sqrt{W}\odot(X-XUV^\top)\| + \lambda(\|XU\|_F^2 + \|V\|_F^2)$, where $XU$ corresponds to the user embedding in WMF. This encourages learning $U$ such that the elements of $XU$ are close to zero. 
    \item Full Rank AWMF. This corresponds to the optimization problem detailed in \cref{eqn:ease}. This model learns a single full rank matrix $B$ and can be seen as the model induced by letting $U$ and $V$ in AWMF be full rank. 
\end{enumerate}
The model numbers in the above list correspond to the numbers in \cref{tab-d1000,tab-d100,tab-d10}. We determine the optimal training hyperparameters $\lambda$ and $\alpha$ by performing a grid search. All models reported in \cref{tab-d1000,tab-d100,tab-d10} were swept over $\lambda \in \{0.0001,0.01,1,100,10\ 000\}$ and $\alpha \in \{1,2,5,10,20\}$, similarly to the sweep used by \citet{liang2018variational} when computing the optimal hyperparameters for WMF. If the highest validation performance occurred at extreme values of the hyperparameters, i.e. $\lambda \in \{0.0001,10\ 000\}$ or $\alpha = 20$, then we expanded the grid until validation performance dropped. All models in \cref{tab-d1000,tab-d100,tab-d10} were implemented using Python, Numpy and Scipy. Models $1,3,4,$ and $5$ in \cref{tab-d1000} were implementing using \cref{alg:conjugate-gradient} and the preconditioner corresponding to their respective unweighted Gram matrices. WMF (model 2) was computed using the publicly available code provided by \citet{liang2018variational}\footnote{Their code adds one to all the values of $\alpha$ as is consistent with \citet{hu2008collaborative}. This is reflected in \cref{tab-d1000,tab-d100,tab-d10}.}, though we remark that WMF could also be implemented with \cref{alg:conjugate-gradient}. To evaluate the ranking accuracy of the various learned models on the withheld test set, we employed the evaluation scheme given by \citet{liang2018variational} as well as the normalized Discounted Cumulative Gain (nDCG$@100$) and Recall ($@20$ and $@50$) detailed therein. 

\subsection{Experimental Results} 
\Cref{tab-d1000,tab-d100,tab-d10} give the ranking accuracies obtained by various models with matrix ranks $d=1000,100$ and $10$, respectively, and for the various regularization and weighting schemes detailed in our paper. \Cref{tab-d1000} also contains the ranking accuracies for the full rank model where $d = |\mathcal{I}|$. 

First, let us discuss the main finding of our paper: that large linear models trained with unweighted data perform comparably to -- and sometimes outperform -- the best performing weighted methods. Specifically, we find that unweighted methods with large model rank either get the best performance or have performance within the standard error of the best performing weighted method. We can see from \cref{tab-d10} that weighting generally benefits models with lower rank. However, as the model rank gets larger, \cref{tab-d1000} shows that weighting begins to hurt model performance in terms of the ranking accuracies. For the full rank model, we see that the unweighted model either beats or nearly matches the performance (Netflix Recall@50 is outside the standard error) of its weighted variant across all three datasets and ranking metrics. We also note that $\alpha=2$ (the smallest weighting in our grid search) performed best for the weighted full rank model. That the ranking accuracy increases with model rank is consistent with the findings of \citet{steck_autoencoders}, who highlight a similar trend for unweighted models. 

The results on MSD might appear to be an outlier at first glance, as all the weighted low rank models outperform the unweighted ones. However, the performance of the full rank model provides a fuller picture: i) The \textit{unweighted} full rank model outperforms the best low rank model (WMF with $d=1000$) by a wide margin; ii) The significantly more rapid performance degradation with smaller $d$ suggests that the underlying structure of the MSD dataset is much more complex than the other two. This is in line with the observation of \citet{steck2021regularization} and \citet{steck2021negative} that MSD has a much longer-tailed distribution. This is also consistent with our observation that weighting matters less when the capacity of a low rank model is increased towards that of a full rank one. 

Now let us turn our attention to the different $\ell_2$ regularizers. Our findings corroborate those of \citet{steck_autoencoders} and \citet{steck2021regularization}, namely that dropout (line 3 \cref{tab-d1000}) seems to outperform weight decay (line 1 \cref{tab-d1000}) in terms of ranking accuracy for the unweighted models with large ranks. When the rank is smaller, as in \cref{tab-d100,tab-d10}, our findings seems to suggest that neither method is conclusively better. Specifically, when the rank is $d=100$, dropout outperforms weight decay on Netflix but underperforms it on ML-20M. 

The hybrid regularizer given in line 4 of the tables was chosen to help elucidate why weighting benefits WMF. As mentioned earlier, if we let $XU$ be the user embedding matrix in WMF then we end up with the hybrid regularizer, which encourages learning a $U$ such that $XU\approx0$. Thus we expect upweighting the nonnegative elements of $X$ to benefit this method, as we do not want to predict zero for the nonnegative entries of $X$. For this method, the best choices of $\lambda$ were values very close to zero, i.e., $10^{-6}$ and $10^{-8}$. This is expected as we do not want to predict zero when $X_{ui}=1$. As $\lambda$ took larger values, it was often very beneficial to the method to use larger weights to compensate for the larger values of $\lambda$. We believe these experimental findings corroborate our intuition that weighting benefited WMF precisely because WMF was implicitly regularizing the user embedding to predict zero when $X$ took nonnegative values. 

In summary, we find that, as the model rank of AWMF grows, the effects of upweighting the positive values of $X$ are detrimental to the performance of the model, as measured by ranking-accuracy. Weighting does indeed help WMF, which is in accordance with the previous literature \citep{hu2008collaborative,pan2008,steck2010,saito2020unbiased}--- and we believe our experimental findings for AWMF with data/weight decay (see previous paragraph) give some intuition into why this phenomena occurs. 
\section{Related Work}

Weighted matrix factorization for implicit feedback data was introduced by \citet{hu2008collaborative} and \citet{pan2008}. 
Based on these seminal works, in many practical applications, the observed/positive user-item interactions were upweighted relative to the unobserved/missing user-item interactions, all with the same weight, resulting in a single scalar weight as an additional hyperparameter that can be tuned during training by optimizing the (unweighted) ranking metric on the validation set. Over the years, it has become the common belief that upweighting the positive user-item interactions is key to obtaining improved recommendation accuracy. This belief was further bolstered by multiple intuitive interpretations of the benefit of upweighting the positive user-item interactions
: (i) it corresponds to reduced uncertainty about the positive data points, and (ii) due to the sparsity of the positive data points, up-weighting them makes the data more balanced, a common practice when dealing with unbalanced data sets. 
Instead of upweighting the positive samples, the negative data-points are down-sampled for computational efficiency in practice. The recent paper of \citet{petrov2023overconfidence} proposes a combination of sampling and weighting.

The matrix factorization model was typically trained iteratively, either by stochastic gradient descent or alternating least squares (e.g. \citet{hu2008collaborative,pan2008, rendle2022revisiting}). As we noticed in the current work, finding a near-optimal solution is not trivial even for such bi-linear models, but it is crucial for better understanding the effect of weighting the data, as shown in this paper. 

While the matrix factorization model decomposes the user-item interaction matrix into a product of low-rank user and item latent factors, similar to (unweighted) singular value decomposition (or pureSVD in the recommendation literature \citep{cremonesi2010puresvd}), other variants were proposed in the literature as well, most notably asymmetric MF \citep{paterek2007improving}, wherein two (different) item-factor matrices are learned, reducing the number of model-parameters considerably (in case of $|\mathcal{U}| \gg |\mathcal{I}|$). Training this model with a weighting scheme has also typically been found to be beneficial.


In contrast, the full rank model called SLIM \citep{ning2011slim}, and its simplified variant called EASE \citep{steck2019embarrassingly} are trained on unweighted data, yet obtain competitive results. These models also use a different variant of $\ell_2$-norm regularization, more closely related to dropout/denoising instead of weight decay \citep{steck_autoencoders}. 

Finally, it is argued by \citet{byrd2019effect} that the effect of weighting diminishes as the capacity of the (deep learning) model increases to the point that it is able to fit all the data, e.g. in a classification task where the model has enough capacity to separate the data. Interestingly, this paper also argues that weight decay regularization may prevent large-norm solutions, which hence may require weighted training---in contrast, dropout style regularization does not have this effect, and hence may require less/no weighting. All of this related work indicates that there is more to uncover regarding the benefits of weighting, and how they are affected by model capacity as well as the regularization, which motivated this paper.

\section{Conclusion}

Our systematic study of weighting schemes for matrix factorization models on implicit feedback data revealed several key findings, which hinge upon the  efficient algorithms derived in this paper for optimizing various weighted objectives. 

The observed experimental results challenge the conventional wisdom of upweighting the observed interactions:
 the benefits of weighting diminish with increasing model rank, and, as model capacity becomes sufficiently large, unweighted methods often outperform their weighted counterparts.
Apart from that,  the choice of regularization scheme interacts with weighting, with dropout-style regularization generally outperforming weight decay for unweighted high rank models, as also found by \citet{steck_autoencoders}. 
 These findings have important implications for recommender system design, suggesting that practitioners should carefully evaluate the necessity of weighting for their specific models and datasets.

Future work could explore the theoretical foundations of these observations and investigate their applicability to deep learning models for recommendation. 
\bibliographystyle{ACM-Reference-Format}
\balance
\bibliography{sample-base}


\begin{thebibliography}{33}


\ifx \showCODEN    \undefined \def \showCODEN     #1{\unskip}     \fi
\ifx \showDOI      \undefined \def \showDOI       #1{#1}\fi
\ifx \showISBNx    \undefined \def \showISBNx     #1{\unskip}     \fi
\ifx \showISBNxiii \undefined \def \showISBNxiii  #1{\unskip}     \fi
\ifx \showISSN     \undefined \def \showISSN      #1{\unskip}     \fi
\ifx \showLCCN     \undefined \def \showLCCN      #1{\unskip}     \fi
\ifx \shownote     \undefined \def \shownote      #1{#1}          \fi
\ifx \showarticletitle \undefined \def \showarticletitle #1{#1}   \fi
\ifx \showURL      \undefined \def \showURL       {\relax}        \fi
\providecommand\bibfield[2]{#2}
\providecommand\bibinfo[2]{#2}
\providecommand\natexlab[1]{#1}
\providecommand\showeprint[2][]{arXiv:#2}

\bibitem[Axelsson(1972)]%
        {axelsson1972generalized}
\bibfield{author}{\bibinfo{person}{Owe Axelsson}.} \bibinfo{year}{1972}\natexlab{}.
\newblock \showarticletitle{A generalized SSOR method}.
\newblock \bibinfo{journal}{\emph{BIT Numerical Mathematics}} (\bibinfo{year}{1972}).
\newblock


\bibitem[Ayoub et~al\mbox{.}(2024)]%
        {ayoubswitching}
\bibfield{author}{\bibinfo{person}{Alex Ayoub}, \bibinfo{person}{Kaiwen Wang}, \bibinfo{person}{Vincent Liu}, \bibinfo{person}{Samuel Robertson}, \bibinfo{person}{James McInerney}, \bibinfo{person}{Dawen Liang}, \bibinfo{person}{Nathan Kallus}, {and} \bibinfo{person}{Csaba Szepesvari}.} \bibinfo{year}{2024}\natexlab{}.
\newblock \showarticletitle{Switching the loss reduces the cost in batch reinforcement learning}. In \bibinfo{booktitle}{\emph{International Conference on Machine Learning (ICML)}}.
\newblock


\bibitem[Bennett and Lanning(2007)]%
        {bennett2007netflix}
\bibfield{author}{\bibinfo{person}{J. Bennett} {and} \bibinfo{person}{S. Lanning}.} \bibinfo{year}{2007}\natexlab{}.
\newblock \showarticletitle{The Netflix Prize}. In \bibinfo{booktitle}{\emph{Proceedings of the KDD Cup Workshop 2007}}.
\newblock


\bibitem[Bertin-Mahieux et~al\mbox{.}(2011)]%
        {msd}
\bibfield{author}{\bibinfo{person}{Thierry Bertin-Mahieux}, \bibinfo{person}{Daniel~P.W. Ellis}, \bibinfo{person}{Brian Whitman}, {and} \bibinfo{person}{Paul Lamere}.} \bibinfo{year}{2011}\natexlab{}.
\newblock \showarticletitle{The Million Song Dataset}. In \bibinfo{booktitle}{\emph{{International Conference on Music Information Retrieval ({ISMIR})}}}.
\newblock


\bibitem[Byrd and Lipton(2019)]%
        {byrd2019effect}
\bibfield{author}{\bibinfo{person}{Jonathon Byrd} {and} \bibinfo{person}{Zachary Lipton}.} \bibinfo{year}{2019}\natexlab{}.
\newblock \showarticletitle{What is the effect of importance weighting in deep learning?}. In \bibinfo{booktitle}{\emph{International conference on machine learning}}. PMLR, \bibinfo{pages}{872--881}.
\newblock


\bibitem[Cavazza et~al\mbox{.}(2018)]%
        {cavazza2018dropout}
\bibfield{author}{\bibinfo{person}{Jacopo Cavazza}, \bibinfo{person}{Pietro Morerio}, \bibinfo{person}{Benjamin Haeffele}, \bibinfo{person}{Connor Lane}, \bibinfo{person}{Vittorio Murino}, {and} \bibinfo{person}{Rene Vidal}.} \bibinfo{year}{2018}\natexlab{}.
\newblock \showarticletitle{Dropout as a Low-Rank Regularizer for Matrix Factorization}. In \bibinfo{booktitle}{\emph{International Conference on Artificial Intelligence and Statistics (AISTATS)}}.
\newblock


\bibitem[Cremonesi et~al\mbox{.}(2010)]%
        {cremonesi2010puresvd}
\bibfield{author}{\bibinfo{person}{Paolo Cremonesi}, \bibinfo{person}{Yedhuda Y.~Koren}, {and} \bibinfo{person}{Roberto Turrin}.} \bibinfo{year}{2010}\natexlab{}.
\newblock \showarticletitle{Performance of Recommender Algorithms on top-{N} recommendation tasks}. In \bibinfo{booktitle}{\emph{ACM Conference on Recommender Systems}}. \bibinfo{pages}{39--46}.
\newblock


\bibitem[De~Pauw and Goethals(2024)]%
        {depauw24}
\bibfield{author}{\bibinfo{person}{Joey De~Pauw} {and} \bibinfo{person}{Bart Goethals}.} \bibinfo{year}{2024}\natexlab{}.
\newblock \showarticletitle{The role of unknown interactions in implicit matrix factorization — A probabilistic view}. In \bibinfo{booktitle}{\emph{ACM Conference on Recommender Systems (RecSys)}}.
\newblock


\bibitem[Deshpande and Karypis(2004)]%
        {deshpande2004item}
\bibfield{author}{\bibinfo{person}{Mukund Deshpande} {and} \bibinfo{person}{George Karypis}.} \bibinfo{year}{2004}\natexlab{}.
\newblock \showarticletitle{Item-based top-n recommendation algorithms}.
\newblock \bibinfo{journal}{\emph{ACM Transactions on Information Systems (TOIS)}} \bibinfo{volume}{22}, \bibinfo{number}{1} (\bibinfo{year}{2004}), \bibinfo{pages}{143--177}.
\newblock


\bibitem[Ferrari~Dacrema et~al\mbox{.}(2019)]%
        {ferrari2019we}
\bibfield{author}{\bibinfo{person}{Maurizio Ferrari~Dacrema}, \bibinfo{person}{Paolo Cremonesi}, {and} \bibinfo{person}{Dietmar Jannach}.} \bibinfo{year}{2019}\natexlab{}.
\newblock \showarticletitle{Are we really making much progress? A worrying analysis of recent neural recommendation approaches}. In \bibinfo{booktitle}{\emph{Proceedings of the 13th ACM conference on recommender systems}}. \bibinfo{pages}{101--109}.
\newblock


\bibitem[Harper and Konstan(2015)]%
        {ml-20m}
\bibfield{author}{\bibinfo{person}{F.~Maxwell Harper} {and} \bibinfo{person}{Joseph~A. Konstan}.} \bibinfo{year}{2015}\natexlab{}.
\newblock \showarticletitle{The MovieLens datasets: History and context}.
\newblock \bibinfo{journal}{\emph{ACM Trans. Interact. Intell. Syst.}} (\bibinfo{year}{2015}).
\newblock


\bibitem[Hestenes and Stiefel(1952)]%
        {cgmethod}
\bibfield{author}{\bibinfo{person}{M.~R. Hestenes} {and} \bibinfo{person}{E. Stiefel}.} \bibinfo{year}{1952}\natexlab{}.
\newblock \showarticletitle{Methods of conjugate gradients for solving linear systems}.
\newblock \bibinfo{journal}{\emph{J. Res. Nat. Bur. Standards}} (\bibinfo{year}{1952}).
\newblock


\bibitem[Horn and Johnson(1991)]%
        {Horn_Johnson_1991}
\bibfield{author}{\bibinfo{person}{Roger~A. Horn} {and} \bibinfo{person}{Charles~R. Johnson}.} \bibinfo{year}{1991}\natexlab{}.
\newblock \bibinfo{booktitle}{\emph{Topics in matrix analysis}}.
\newblock \bibinfo{publisher}{Cambridge University Press}.
\newblock


\bibitem[Hu et~al\mbox{.}(2008)]%
        {hu2008collaborative}
\bibfield{author}{\bibinfo{person}{Yifan Hu}, \bibinfo{person}{Yehuda Koren}, {and} \bibinfo{person}{Chris Volinsky}.} \bibinfo{year}{2008}\natexlab{}.
\newblock \showarticletitle{Collaborative filtering for implicit feedback datasets}. In \bibinfo{booktitle}{\emph{International Conference on Data Mining}}.
\newblock


\bibitem[J{\"a}rvelin and Kek{\"a}l{\"a}inen(2000)]%
        {jarvelin2000ir}
\bibfield{author}{\bibinfo{person}{Kalervo J{\"a}rvelin} {and} \bibinfo{person}{Jaana Kek{\"a}l{\"a}inen}.} \bibinfo{year}{2000}\natexlab{}.
\newblock \showarticletitle{IR evaluation methods for retrieving highly relevant documents}. In \bibinfo{booktitle}{\emph{International ACM SIGIR Conference on Research and Development in Information Retrieval}}.
\newblock


\bibitem[Jin et~al\mbox{.}(2021)]%
        {jin2021towards}
\bibfield{author}{\bibinfo{person}{Ruoming Jin}, \bibinfo{person}{Dong Li}, \bibinfo{person}{Jing Gao}, \bibinfo{person}{Zhi Liu}, \bibinfo{person}{Li Chen}, {and} \bibinfo{person}{Yang Zhou}.} \bibinfo{year}{2021}\natexlab{}.
\newblock \showarticletitle{Towards a better understanding of linear models for recommendation}. In \bibinfo{booktitle}{\emph{ACM SIGKDD Conference on Knowledge Discovery \& Data Mining}}.
\newblock


\bibitem[Liang et~al\mbox{.}(2016)]%
        {liang2016modeling}
\bibfield{author}{\bibinfo{person}{Dawen Liang}, \bibinfo{person}{Laurent Charlin}, \bibinfo{person}{James McInerney}, {and} \bibinfo{person}{David~M Blei}.} \bibinfo{year}{2016}\natexlab{}.
\newblock \showarticletitle{Modeling user exposure in recommendation}. In \bibinfo{booktitle}{\emph{Proceedings of the 25th international conference on World Wide Web}}. \bibinfo{pages}{951--961}.
\newblock


\bibitem[Liang et~al\mbox{.}(2018)]%
        {liang2018variational}
\bibfield{author}{\bibinfo{person}{Dawen Liang}, \bibinfo{person}{Rahul~G Krishnan}, \bibinfo{person}{Matthew~D Hoffman}, {and} \bibinfo{person}{Tony Jebara}.} \bibinfo{year}{2018}\natexlab{}.
\newblock \showarticletitle{Variational autoencoders for collaborative filtering}. In \bibinfo{booktitle}{\emph{The World Wide Web Conference (WWW)}}.
\newblock


\bibitem[Ning and Karypis(2011)]%
        {ning2011slim}
\bibfield{author}{\bibinfo{person}{Xia Ning} {and} \bibinfo{person}{George Karypis}.} \bibinfo{year}{2011}\natexlab{}.
\newblock \showarticletitle{Slim: Sparse linear methods for top-n recommender systems}. In \bibinfo{booktitle}{\emph{International Conference on Data Mining (ICDM)}}.
\newblock


\bibitem[Nocedal and Wright(1999)]%
        {nocedal1999numerical}
\bibfield{author}{\bibinfo{person}{Jorge Nocedal} {and} \bibinfo{person}{Stephen~J Wright}.} \bibinfo{year}{1999}\natexlab{}.
\newblock \bibinfo{booktitle}{\emph{Numerical optimization}}.
\newblock \bibinfo{publisher}{Springer}.
\newblock


\bibitem[Pan et~al\mbox{.}(2008)]%
        {pan2008}
\bibfield{author}{\bibinfo{person}{Rong Pan}, \bibinfo{person}{Yunhong Zhou}, \bibinfo{person}{Bin Cao}, \bibinfo{person}{Nathan~N. Liu}, \bibinfo{person}{Rajan Lukose}, \bibinfo{person}{Martin Scholz}, {and} \bibinfo{person}{Qiang Yang}.} \bibinfo{year}{2008}\natexlab{}.
\newblock \showarticletitle{One-class collaborative filtering}. In \bibinfo{booktitle}{\emph{International Conference on Data Mining}}.
\newblock


\bibitem[Paterek(2007)]%
        {paterek2007improving}
\bibfield{author}{\bibinfo{person}{Arkadiusz Paterek}.} \bibinfo{year}{2007}\natexlab{}.
\newblock \showarticletitle{Improving regularized singular value decomposition for collaborative filtering}. In \bibinfo{booktitle}{\emph{Proceedings of KDD cup and workshop}}, Vol.~\bibinfo{volume}{2007}. \bibinfo{pages}{5--8}.
\newblock


\bibitem[Petrov and Macdonald(2023)]%
        {petrov2023overconfidence}
\bibfield{author}{\bibinfo{person}{Aleksandr~Vladimirovich Petrov} {and} \bibinfo{person}{Craig Macdonald}.} \bibinfo{year}{2023}\natexlab{}.
\newblock \showarticletitle{g{SASR}ec: Reducing Overconfidence in Sequential Recommendation Trained with Negative Sampling}. In \bibinfo{booktitle}{\emph{ACM Conference on Recommender Systems}}.
\newblock


\bibitem[Rendle et~al\mbox{.}(2009)]%
        {rendle2009bpr}
\bibfield{author}{\bibinfo{person}{Steffen Rendle}, \bibinfo{person}{Christoph Freudenthaler}, \bibinfo{person}{Zeno Gantner}, {and} \bibinfo{person}{Lars Schmidt-Thieme}.} \bibinfo{year}{2009}\natexlab{}.
\newblock \showarticletitle{BPR: Bayesian personalized ranking from implicit feedback}. In \bibinfo{booktitle}{\emph{Proceedings of the Twenty-Fifth Conference on Uncertainty in Artificial Intelligence}}. \bibinfo{pages}{452--461}.
\newblock


\bibitem[Rendle et~al\mbox{.}(2022)]%
        {rendle2022revisiting}
\bibfield{author}{\bibinfo{person}{Steffen Rendle}, \bibinfo{person}{Walid Krichene}, \bibinfo{person}{Li Zhang}, {and} \bibinfo{person}{Yehuda Koren}.} \bibinfo{year}{2022}\natexlab{}.
\newblock \showarticletitle{Revisiting the performance of i{ALS} on item recommendation benchmarks}. In \bibinfo{booktitle}{\emph{Proceedings of the 16th ACM Conference on Recommender Systems}}. \bibinfo{pages}{427--435}.
\newblock


\bibitem[Saito et~al\mbox{.}(2020)]%
        {saito2020unbiased}
\bibfield{author}{\bibinfo{person}{Yuta Saito}, \bibinfo{person}{Suguru Yaginuma}, \bibinfo{person}{Yuta Nishino}, \bibinfo{person}{Hayato Sakata}, {and} \bibinfo{person}{Kazuhide Nakata}.} \bibinfo{year}{2020}\natexlab{}.
\newblock \showarticletitle{Unbiased recommender learning from missing-not-at-random implicit feedback}. In \bibinfo{booktitle}{\emph{International Conference on Web Search and Data Mining}}.
\newblock


\bibitem[Steck(2010)]%
        {steck2010}
\bibfield{author}{\bibinfo{person}{Harald Steck}.} \bibinfo{year}{2010}\natexlab{}.
\newblock \showarticletitle{Training and testing of recommender systems on data missing not at random}. In \bibinfo{booktitle}{\emph{ACM SIGKDD International Conference on Knowledge Discovery and Data Mining (KDD)}}.
\newblock


\bibitem[Steck(2015)]%
        {steck2015gaussian}
\bibfield{author}{\bibinfo{person}{Harald Steck}.} \bibinfo{year}{2015}\natexlab{}.
\newblock \showarticletitle{Gaussian ranking by matrix factorization}. In \bibinfo{booktitle}{\emph{Proceedings of the 9th ACM Conference on Recommender Systems}}. \bibinfo{pages}{115--122}.
\newblock


\bibitem[Steck(2019a)]%
        {steck2019collaborative}
\bibfield{author}{\bibinfo{person}{Harald Steck}.} \bibinfo{year}{2019}\natexlab{a}.
\newblock \showarticletitle{Collaborative filtering via high-dimensional regression}.
\newblock \bibinfo{journal}{\emph{arXiv preprint arXiv:1904.13033}} (\bibinfo{year}{2019}).
\newblock


\bibitem[Steck(2019b)]%
        {steck2019embarrassingly}
\bibfield{author}{\bibinfo{person}{Harald Steck}.} \bibinfo{year}{2019}\natexlab{b}.
\newblock \showarticletitle{Embarrassingly shallow autoencoders for sparse data}. In \bibinfo{booktitle}{\emph{The World Wide Web Conference (WWW)}}.
\newblock


\bibitem[Steck(2020)]%
        {steck_autoencoders}
\bibfield{author}{\bibinfo{person}{Harald Steck}.} \bibinfo{year}{2020}\natexlab{}.
\newblock \showarticletitle{Autoencoders that don\textquotesingle t overfit towards the Identity}. In \bibinfo{booktitle}{\emph{Advances in Neural Information Processing Systems (NeurIPS)}}. \bibinfo{publisher}{Curran Associates, Inc.}
\newblock


\bibitem[Steck and Garcia(2021)]%
        {steck2021regularization}
\bibfield{author}{\bibinfo{person}{Harald Steck} {and} \bibinfo{person}{Dario~Garcia Garcia}.} \bibinfo{year}{2021}\natexlab{}.
\newblock \showarticletitle{On the regularization of autoencoders}.
\newblock \bibinfo{journal}{\emph{arXiv preprint arXiv:2110.11402}} (\bibinfo{year}{2021}).
\newblock


\bibitem[Steck and Liang(2021)]%
        {steck2021negative}
\bibfield{author}{\bibinfo{person}{Harald Steck} {and} \bibinfo{person}{Dawen Liang}.} \bibinfo{year}{2021}\natexlab{}.
\newblock \showarticletitle{Negative interactions for improved collaborative filtering: Don’t go deeper, go higher}. In \bibinfo{booktitle}{\emph{Proceedings of the 15th ACM Conference on Recommender Systems}}. \bibinfo{pages}{34--43}.
\newblock


\end{thebibliography}
\FloatBarrier
\appendix
\clearpage
\onecolumn
\section{Tables Grouped by Dataset}

In this section we present the results given in \cref{sec:numerical-experiments} organized by dataset instead of model rank.

\begin{table*}[h]
\caption{Ranking accuracies on ML-20M for various linear models. Columns correspond to model rank \(d\in\{10,100,1000\}\). For each model the top row shows the weighted variant (\(\alpha\neq1\)) and the bottom row the unweighted variant (\(\alpha=1\)); the “Model” column shows the corresponding equation.}
\label{tab-ML20M-reorg}
\resizebox{\textwidth}{!}{
\begin{tabular}{lrrrrcrrrcrrrcc}
 & & \multicolumn{4}{c}{{\bf$d=10$}} & \multicolumn{4}{c}{{\bf $d=100$}} & \multicolumn{4}{c}{{\bf $d=1000$}} \\
\cmidrule(lr){3-6} \cmidrule(lr){7-10} \cmidrule(lr){11-14} \cmidrule(lr){15-15}
 & & Recall & Recall & nDCG & \multirow{2}{*}{$\alpha$} & Recall & Recall & nDCG & \multirow{2}{*}{$\alpha$} & Recall & Recall & nDCG & \multirow{2}{*}{$\alpha$}   \\
\multicolumn{2}{l}{{\bf Model}} & @20 \hfill{} & @50 \hfill{} & @100\hfill{} & & @20 \hfill{} & @50 \hfill{} & @100 \hfill{} & & @20 \hfill{}   & @50 \hfill{} & @100   \\
\cmidrule(lr){1-2} \cmidrule(lr){3-6} \cmidrule(lr){7-10} \cmidrule(lr){11-14} \cmidrule(lr){15-15}
\multirow{2}{*}{1} 
  & \(\n{\sqrt W\odot (X - XUV\t)} + \lambda(\n{U} + \n{V})\) 
  & 0.305 & 0.426 & 0.338 & 5 
  & 0.339 & 0.472 & 0.374 & 2 
  & 0.344 & 0.476 & 0.378 & 2 \\
  & \(\n{X - XUV\t} + \lambda(\n{U} + \n{V})\) 
  & 0.297 & 0.415 & 0.332 & 1 
  & 0.348 & 0.466 & 0.381 & 1 
  & 0.348 & 0.466 & 0.381 & 1 \\[2mm]
\multirow{2}{*}{2} 
  & \(\n{\sqrt W\odot (X - UV\t)} + \lambda(\n{U} + \n{V})\)
  & 0.305 & 0.427 & 0.339 & 3 
  & 0.364 & 0.497 & 0.399 & 6 
  & 0.377 & 0.512 & 0.409 & 3 \\
  & \(\n{X - UV\t} + \lambda(\n{U} + \n{V})\)
  & 0.289 & 0.401 & 0.322 & 1 
  & 0.319 & 0.434 & 0.350 & 1 
  & 0.328 & 0.443 & 0.360 & 1 \\[2mm]
\multirow{2}{*}{3} 
  & \(\n{\sqrt W\odot (X - XUV\t)} + \lambda\n{UV\t}\)
  & 0.303 & 0.425 & 0.336 & 5 
  & 0.320 & 0.454 & 0.361 & 2 
  & 0.355 & 0.491 & 0.388 & 2 \\
  & \(\n{X - XUV\t} + \lambda\n{UV\t}\)
  & 0.298 & 0.415 & 0.331 & 1 
  & 0.336 & 0.466 & 0.368 & 1 
  & 0.378 & 0.511 & 0.407 & 1 \\[2mm]
\multirow{2}{*}{4} 
  & \(\n{\sqrt W\odot (X - XUV\t)} + \lambda(\n{XU} + \n{V})\)
  & 0.293 & 0.412 & 0.324 & 2 
  & 0.344 & 0.475 & 0.374 & 2 
  & 0.360 & 0.499 & 0.389 & 2 \\
  & \(\n{X - XUV\t} + \lambda(\n{XU} + \n{V})\)
  & 0.290 & 0.407 & 0.321 & 1 
  & 0.315 & 0.430 & 0.346 & 1 
  & 0.376 & 0.511 & 0.407 & 1 \\[2mm]
\midrule
\multirow{2}{*}{5} 
  & Full Rank: \(\n{\sqrt W\odot (X - XB)} + \lambda(\n{B})\)
  & ---   & ---   & ---   & --- 
  & ---   & ---   & ---   & --- 
  & 0.360 & 0.499 & 0.389 & 2 \\
  & \(\n{X - XB} + \lambda(\n{B})\)
  & ---   & ---   & ---   & --- 
  & ---   & ---   & ---   & --- 
  & 0.376 & 0.511 & 0.407 & 1 \\
\bottomrule
\end{tabular}
}
\end{table*}
\begin{table*}[h]
\caption{Ranking accuracies on Netflix.}
\label{tab-Netflix-reorg}
\resizebox{\textwidth}{!}{
\begin{tabular}{lrrrrcrrrcrrrcc}
 & & \multicolumn{4}{c}{{\bf$d=10$}} & \multicolumn{4}{c}{{\bf $d=100$}} & \multicolumn{4}{c}{{\bf $d=1000$}} \\
\cmidrule(lr){3-6} \cmidrule(lr){7-10} \cmidrule(lr){11-14} \cmidrule(lr){15-15}
 & & Recall & Recall & nDCG & \multirow{2}{*}{$\alpha$} & Recall & Recall & nDCG & \multirow{2}{*}{$\alpha$} & Recall & Recall & nDCG & \multirow{2}{*}{$\alpha$}   \\
\multicolumn{2}{l}{{\bf Model}} & @20 \hfill{} & @50 \hfill{} & @100\hfill{} & & @20 \hfill{} & @50 \hfill{} & @100 \hfill{} & & @20 \hfill{}   & @50 \hfill{} & @100   \\
\cmidrule(lr){1-2} \cmidrule(lr){3-6} \cmidrule(lr){7-10} \cmidrule(lr){11-14} \cmidrule(lr){15-15}
\multirow{2}{*}{1} 
  & \(\n{\sqrt W\odot (X - XUV\t)} + \lambda(\n{U} + \n{V})\) 
  & 0.259 & 0.339 & 0.293 & 2 
  & 0.309 & 0.392 & 0.343 & 5 
  & 0.311 & 0.395 & 0.345 & 5 \\
  & \(\n{X - XUV\t} + \lambda(\n{U} + \n{V})\) 
  & 0.258 & 0.335 & 0.290 & 1 
  & 0.302 & 0.376 & 0.334 & 1 
  & 0.302 & 0.376 & 0.334 & 1 \\[2mm]
\multirow{2}{*}{2} 
  & \(\n{\sqrt W\odot (X - UV\t)} + \lambda(\n{U} + \n{V})\)
  & 0.258 & 0.340 & 0.293 & 3 
  & 0.316 & 0.396 & 0.348 & 3 
  & 0.292 & 0.376 & 0.329 & 2 \\
  & \(\n{X - UV\t} + \lambda(\n{U} + \n{V})\)
  & 0.259 & 0.337 & 0.292 & 1 
  & 0.312 & 0.388 & 0.344 & 1 
  & 0.318 & 0.406 & 0.352 & 1 \\[2mm]
\multirow{2}{*}{3} 
  & \(\n{\sqrt W\odot (X - XUV\t)} + \lambda\n{UV\t}\)
  & 0.260 & 0.340 & 0.293 & 2 
  & 0.316 & 0.396 & 0.348 & 2 
  & 0.337 & 0.419 & 0.370 & 2 \\
  & \(\n{X - XUV\t} + \lambda\n{UV\t}\)
  & 0.258 & 0.334 & 0.290 & 1 
  & 0.318 & 0.395 & 0.369 & 1 
  & 0.337 & 0.417 & 0.369 & 1 \\[2mm]
\multirow{2}{*}{4} 
  & \(\n{\sqrt W\odot (X - XUV\t)} + \lambda(\n{XU} + \n{V})\)
  & 0.260 & 0.339 & 0.294 & 2 
  & 0.304 & 0.384 & 0.338 & 2 
  & 0.388 & 0.421 & 0.370 & 2 \\
  & \(\n{X - XUV\t} + \lambda(\n{XU} + \n{V})\)
  & 0.260 & 0.338 & 0.293 & 1 
  & 0.312 & 0.388 & 0.343 & 1 
  & 0.335 & 0.417 & 0.368 & 1 \\[2mm]
  \midrule
\multirow{2}{*}{5} 
  & Full Rank: \(\n{\sqrt W\odot (X - XB)} + \lambda(\n{B})\)
  & ---   & ---   & ---   & --- 
  & ---   & ---   & ---   & --- 
  & 0.338 & 0.421 & 0.370 & 2 \\
  & \(\n{X - XB} + \lambda(\n{B})\)
  & ---   & ---   & ---   & --- 
  & ---   & ---   & ---   & --- 
  & 0.335 & 0.417 & 0.368 & 1 \\
\bottomrule
\end{tabular}
}
\end{table*}

\begin{table*}[h]
\caption{Ranking accuracies on MSD.}
\label{tab-MSD-reorg}
\resizebox{\textwidth}{!}{
\begin{tabular}{lrrrrcrrrcrrrcc}
 & & \multicolumn{4}{c}{{\bf$d=10$}} & \multicolumn{4}{c}{{\bf $d=100$}} & \multicolumn{4}{c}{{\bf $d=1000$}} \\
\cmidrule(lr){3-6} \cmidrule(lr){7-10} \cmidrule(lr){11-14} \cmidrule(lr){15-15}
 & & Recall & Recall & nDCG & \multirow{2}{*}{$\alpha$} & Recall & Recall & nDCG & \multirow{2}{*}{$\alpha$} & Recall & Recall & nDCG & \multirow{2}{*}{$\alpha$}   \\
\multicolumn{2}{l}{{\bf Model}} & @20 \hfill{} & @50 \hfill{} & @100\hfill{} & & @20 \hfill{} & @50 \hfill{} & @100 \hfill{} & & @20 \hfill{}   & @50 \hfill{} & @100   \\
\cmidrule(lr){1-2} \cmidrule(lr){3-6} \cmidrule(lr){7-10} \cmidrule(lr){11-14} \cmidrule(lr){15-15}
\multirow{2}{*}{1} 
  & \(\n{\sqrt W\odot (X - XUV\t)} + \lambda(\n{U} + \n{V})\)
  & 0.086 & 0.134 & 0.114 & 20  
  & 0.167 & 0.241 & 0.210 & 10  
  & 0.228 & 0.314 & 0.277 & 10 \\
  & \(\n{X - XUV\t} + \lambda(\n{U} + \n{V})\)
  & 0.074 & 0.111 & 0.099 & 1  
  & 0.129 & 0.182 & 0.163 & 1  
  & 0.189 & 0.251 & 0.230 & 1 \\[2mm]
\multirow{2}{*}{2} 
  & \(\n{\sqrt W\odot (X - UV\t)} + \lambda(\n{U} + \n{V})\)
  & 0.082 & 0.130 & 0.111 & 21  
  & 0.178 & 0.265 & 0.223 & 21  
  & 0.268 & 0.371 & 0.322 & 21 \\
  & \(\n{X - UV\t} + \lambda(\n{U} + \n{V})\)
  & 0.074 & 0.110 & 0.099 & 1  
  & 0.126 & 0.179 & 0.160 & 1  
  & 0.189 & 0.251 & 0.231 & 1 \\[2mm]
\multirow{2}{*}{3} 
  & \(\n{\sqrt W\odot (X - XUV\t)} + \lambda\n{UV\t}\)
  & 0.085 & 0.129 & 0.111 & 10  
  & 0.167 & 0.239 & 0.208 & 10  
  & 0.227 & 0.311 & 0.277 & 10 \\
  & \(\n{X - XUV\t} + \lambda\n{UV\t}\)
  & 0.075 & 0.110 & 0.098 & 1  
  & 0.132 & 0.186 & 0.167 & 1  
  & 0.216 & 0.296 & 0.266 & 1 \\[2mm]
\multirow{2}{*}{4} 
  & \(\n{\sqrt W\odot (X - XUV\t)} + \lambda(\n{XU} + \n{V})\)
  & 0.083 & 0.129 & 0.110 & 20  
  & 0.171 & 0.247 & 0.213 & 20  
  & 0.231 & 0.322 & 0.283 & 10 \\
  & \(\n{X - XUV\t} + \lambda(\n{XU} + \n{V})\)
  & 0.074 & 0.110 & 0.098 & 1  
  & 0.128 & 0.180 & 0.161 & 1  
  & 0.190 & 0.251 & 0.231 & 1 \\[2mm]
\midrule
\multirow{2}{*}{5} 
  & Full Rank: \(\n{\sqrt W\odot (X - XB)} + \lambda(\n{B})\)
  & ---   & ---   & ---   & ---  
  & ---   & ---   & ---   & ---  
  & 0.277 & 0.377 & 0.338 & 2 \\
  & \(\n{X - XB} + \lambda(\n{B})\)
  & ---   & ---   & ---   & ---  
  & ---   & ---   & ---   & ---  
  & 0.284 & 0.384 & 0.344 & 1 \\[2mm]
\bottomrule
\end{tabular}
}
\end{table*}
\end{document}